\documentclass[lettersize,journal]{IEEEtran}
\usepackage{cite}
\usepackage{amsmath,amssymb,amsfonts,amsthm}
\usepackage{mathrsfs}
\usepackage{array}

\usepackage{tabularx}
\usepackage{multirow}
\usepackage{adjustbox}
\usepackage[dvipsnames]{xcolor}
\usepackage{lipsum}
\usepackage{booktabs}
\usepackage{color,colortbl}
\definecolor{lavender}{rgb}{0.9, 0.9, 0.98}
\usepackage{mathtools}
\usepackage{subfigure}
\usepackage{mathtools}
\usepackage{graphicx}

\usepackage{caption}
\usepackage{subcaption}

\usepackage{bbm}
\usepackage[utf8]{inputenc}
\usepackage{soul}

\usepackage{pifont}

\usepackage{algpseudocode}
\usepackage[linesnumbered,ruled,vlined]{algorithm2e}
\SetKwInput{KwInput}{Input}                
\SetKwInput{KwOutput}{Output}
\SetKwInput{KwInit}{Initialization}
\SetKwInput{KwEndFor}{end for}

\newtheorem{theorem}{Theorem}

\newtheorem{remark}{Remark}
\newtheorem{corollary}{Corollary}[theorem]

\usepackage{textcomp}

\usepackage[normalem]{ulem}
\allowdisplaybreaks

\begin{document}

\title{Capacity Analysis of Cascaded BD-RIS Assisted MIMO Systems}
\author{M. S. S. Manasa,  Praful D. Mankar, and Sundaram Vanka
\thanks{M. S. S. Manasa, and P. D. Mankar are with Signal Processing and Communication Research Center, IIIT Hyderabad, India. Email: mss.manasa@research.iiit.ac.in, praful.mankar@iiit.ac.in. \newline S. Vanka is with department of electrical engineering at Indian Institute of Technology, Hyderabad. Email: sundar.vanka@ee.iith.ac.in} 
}

\maketitle

\begin{abstract}
This paper examines the cascaded deployment of beyond diagonal (BD) reconfigurable intelligent surfaces (RISs) and explores its potential to enhance the performance of MIMO systems.
We first derive the jointly optimal closed-form solutions for the RISs in cascade with SVD–water-filling (SVD–WF) and uniform power allocation (UPA) precoding strategies.
The optimally configured cascaded-RIS with UPA is shown to achieve performance comparable to that with the SVD–WF approach, suggesting that cascaded-RISs can also aid in reducing transmitter complexity.
Furthermore, the approximate ergodic capacity for UPA is derived, along with its high-SNR approximation which provides multiple useful insights into the dimension and deployment of cascaded RISs. 
The analytical results establish a clear tradeoff among transmit power, RIS size, and achievable capacity, providing insights for practical deployment in high-SNR cascaded-RIS MIMO systems. 
\end{abstract}


\section{Introduction}\label{sec : Intro}
Reconfigurable Intelligent Surfaces (RIS) have emerged as a transformative technology in wireless communications, offering control over the propagation environment. These surfaces consist of a large array of reflecting elements which can be dynamically configured to manipulate the phase, amplitude, or polarization of incident EM waves, thereby enabling intelligent signal redirection \cite{basar2019wireless,wu2019towards}. 
Thus, RIS provides an attractive solution to meet the coverage and capacity requirements of future wireless networks. 
This paradigm shift has led to a surge in research focused on multi-RIS-assisted MIMO systems, beamforming optimization and the capacity analysis, which also is the focus of this paper.
\subsubsection{RIS Architectures --- Diagonal and Beyond-Diagonal}
The originally conceptualized design of RIS comprised of independently tunable reflecting elements modeled by a diagonal matrix of reflection coefficients \cite{Diagonal_RIS , Diagonal_RIS_survey}, and thus termed {\em diagonal RIS} (D-RIS). While such RISs are energy-efficient and low-cost, they suffer from (i) unit-modulus constraints which complicate optimization, especially in MU-MIMO systems \cite{zhang2021capacity}, and (ii) limited degrees of freedom (DoF) due to the lack of inter-element interaction \cite{BD_RIS}. To address these limitations, a new design has emerged where elements can be controlled jointly \cite{BD_RIS_Intro}. Due to the inter-element coupling, the RIS response is represented by a non-diagonal full or block matrix, leading to the concept of {\em beyond-diagonal RISs (BD-RISs)} with a unitary constraint. This additional DoF in configuring RIS enhances flexibility in EM wave manipulation and improves system capacity \cite{BD_RIS_Intro_to_optimization}. The unitary constraint of phase shift matrix enables low-complexity optimal phase design, as shown in \cite{emil2025capacity} and is further explored in this paper.
\subsubsection{RIS Deployment --  Multi-RIS Systems}
Beyond single-RIS systems, multi-RIS deployments have gained attention for improving diversity and coverage wherein the two key RIS arrangements are parallel or cascade \cite{parallel_cascaded_RIS}. Parallel RISs enhance spatial DoF. Cascaded RISs extend coverage and array gain which improves end-to-end SNR \cite{Cascade_RIS}. Ultimately, cascaded arrangement shall be the focus of this paper. However, existing works on multi-RIS systems are limited to D-RISs, leaving multi BD-RIS systems unexplored despite their potential to achieve high performance with reduced complexity. 
\subsubsection{Optimal Configuration for RIS systems}
A key research direction in RIS-assisted MIMO systems is the joint optimization of transceiver precoding and RIS phase shifts to maximize utility functions such as capacity, energy efficiency, etc. These optimization formulations, often non-convex and challenging to solve, are fundamentally influenced by the constraints imposed by the RIS architectures  \cite{RIS_arch1}. There have been many efforts in literature on optimizing the single/multi D-RIS aided systems. However, these works mostly provide iterative solutions with high computational complexity, with the exception of \cite{zhang2021capacity} whose authors derived closed form expressions for RIS phase shifts that still needs to be solved iteratively. 
In contrast, there have been fewer efforts to optimize BD-RIS systems. In \cite{joint_ActivePassive_BD-RIS,BD_RIS_Intro_to_optimization}, fractional programming based frameworks are developed, whereas \cite{emil2025capacity} derives a closed form expression for optimal single BD-RIS phase shift matrix.
However, the optimization of multi-BD-RIS  systems is yet to be explored. Finally, because of the lack of tractable closed-form optimal solutions for phase shift matrices, the characterization of ergodic capacity (EC) for D/BD single/multi RIS-aided MIMO system still remains an open problem.

{\em Contributions:} Motivated by the above gap, {we focus on finding the optimal RIS phase shift matrices for cascaded BD-RIS systems that ensure low computational complexity solutions and also aid in analytical tractability for EC analysis.}
The main contributions are summarized as follows:
\begin{itemize}
    \item Closed-form expressions are derived for the optimal phase shift matrices for cascaded RISs-MIMO system.
    
    \item The proposed optimal RIS solution with UPA precoding is shown to attain performance comparable to that with SVD–WF precoding, simplifying the transmitter design.
   
   \item An approximate closed-form expression for EC is derived which is simplified in high-SNR to reveal how capacity scales with the number of cascaded RISs and their dimensions.
\end{itemize}

\section{System Model}\label{Sys_Model}
This work considers a multi-RIS assisted MIMO communication system where the RISs are placed in {\em cascade} arrangement between the transmitter and users. 
This setting consists of a base station (BS) equipped  with $M$ transmit antennas, $K$ single-antenna users, and $L$ RISs placed sequentially between the BS and users such that the $l$-th RIS has $N_l$ number of reflecting elements. 
This setup includes  $L+1$ sequential links: {\rm BS--RIS$_1$}, {\rm RIS$_1$--RIS$_2$}, \dots, {\rm RIS$_{\rm L-1}$--RIS$_{\rm L}$}, and {\rm RIS$_{\rm L}$--Users}.
The phase configuration of $l$-th RIS is denoted by an $N_l\times N_l$ matrix  $\mathbf{\Phi}_l$ such that $\mathbf{\Phi}_l\mathbf{\Phi}_l^{\rm H}=\mathbf{I}_{N_l}$ where $\mathbf{I}_{N_l}$ is the identity matrix.
Let the channels from the BS to RIS$_1$ be denoted as $\mathbf{H}_{0} \in \mathbb{C}^{N_1 \times M}$, from RIS$_{l}$ to RIS$_{l+1}$ be $\mathbf{H}_{l} \in \mathbb{C}^{N_{l+1} \times N_l}$, and from RIS$_{\rm L}$ to the users be $\mathbf{H}_{\rm L} \in \mathbb{C}^{K \times N_L}$. Thus, the composite cascaded end-to-end channel can be written as
\begin{equation}
    \mathbf{H}_{\rm cas} = \mathbf{H}_{\rm L}\mathbf{\Phi}_{\rm L}\mathbf{H}_{\rm L-1} \mathbf{\Phi}_{\rm L-1} \cdots \mathbf{H}_{2} \mathbf{\Phi}_2 \mathbf{H}_1 \mathbf{\Phi}_1 \mathbf{H}_{0}.
\end{equation}
Let $\mathbf{x} \in \mathbb{C}^{M \times 1}$ be the transmitted signal vector defined as
\begin{equation}
    \mathbf{x} = \sum_i \sqrt{p_i} \mathbf{v}_i s_i,
\end{equation}
where $s_i$ is the modulation symbol (such that $\mathbb{E}[|s_i|^2] = 1$), $\mathbf{v}_i \in \mathbb{C}^{M \times 1}$ is the precoding vector (such that $\|\mathbf{v}_i\|_2 = 1 $), and $p_i$ is the transmit power allocated to $i$-th symbol.
Thus, the received signal vector can be written as
\begin{equation}
    \mathbf{y} =  \mathbf{H}_{\rm cas}\mathbf{x} + \mathbf{n}, \label{eq:series_received}
\end{equation}
where $\mathbf{n} \in \mathbb{C}^{K \times 1}$ is an additive white Gaussian noise (AWGN) vector with zero mean and covariance matrix $\sigma_n^2 \mathbf{I}_{\rm K}$. 
The capacity under the cascade arrangement having composite cascaded channel $\mathbf{H}_{\rm cas}$ is
\begin{equation}
    \mathrm{C}_{\rm cas} = \log_2|\mathbf{I}_{\rm K} + \mathbf{H}_{\rm cas}\mathbf{QH}_{\rm cas}^{\rm H}|,\label{eq:MIMO_Capacity_cas}
\end{equation}
where $\mathbf{Q} = \mathbb{E}[\mathbf{xx}^{\rm H}]$ is the transmit precoding matrix. Assuming perfect knowledge of CSI, the aim is to appropriately configure the RIS phase shift matrices $\{\mathbf{\Phi}_1,\mathbf{\Phi}_2,\ldots,\mathbf{\Phi}_{\rm L}\}$ and precoding matrix $\mathbf{Q}$ such that the cascaded transmission capacity is maximized. 
Thus, we formulate the cascaded capacity maximization problem as
\begin{subequations}
\begin{align}
\textbf{(P1)}:
\max_{\stackrel{\boldsymbol{\Phi}_{\rm 1},\ldots,\boldsymbol{\Phi}_{\rm L}}{\mathbf{Q}}} ~~& \log_2|\mathbf{I}_{\rm K} + \mathbf{H}_{\rm cas}\mathbf{QH}_{\rm cas}^{\rm H}|,\label{eq:Capacity_objective_cas}\\
\textit{\emph{s.t}} ~~& \mathrm{Tr}(\mathbf{Q}) \leq P_t,\label{eq:power_constraint_cas}\\
\textit{\emph{and}} ~~&  \boldsymbol{\Phi}_l\boldsymbol{\Phi}_l^{\rm H} = \mathbf{I}_{N_l}, ~~~\forall l=1,\dots, L.\label{eq:Phi_constraint_cas}
\end{align}\label{eq:Objective_cas}
\end{subequations}
where \eqref{eq:power_constraint_cas} limits the transmission power and \eqref{eq:Phi_constraint_cas} ensures that the phase shift matrices for BD-RISs are unitary. Here, the maximum achievable capacity $\mathrm{C}^*_{\rm cas}$ obtained using the optimal RIS configuration and transmit precoding corresponds to a given channel realization.
We derive accurate closed-form expressions for the  characterization of EC $\mathrm{EC}_{\rm cas}=\mathbb{E}[\mathrm{C}^*_{\rm cas}]$ and utilize them to gain deeper insights into the performance behavior of cascaded RIS architectures.

\section{ Optimal RIS Phase shift matrix}
The capacity maximization problem presented in \eqref{eq:Objective_cas} is equivalent to a determinant maximization problem and can be expressed as maximization of $\prod_{i=1}^{R}\Lambda_i$ subjected to constraints \eqref{eq:power_constraint_cas} and \eqref{eq:Phi_constraint_cas}, 
where $R$ represents the rank of the matrix $\mathbf{H}_{\rm cas }\mathbf{Q}\mathbf{H}_{\rm cas}^{\rm H}$ and $\Lambda_i$s are its eigenvalues. This quantity is maximized when all the eigenvalues are equal i.e. $\Lambda_i=\Lambda_j$ for $\forall i,j$ under feasible solution set. Thus, because of AM-GM inequality \cite{amgm}, the optimal solution should satisfy 
\begin{equation}
    |\mathbf{H}_{\rm cas }\mathbf{Q}\mathbf{H}_{\rm cas}^{\rm H}|=\frac{1}{R}{\rm Tr}(\mathbf{H}_{\rm cas }\mathbf{Q}\mathbf{H}_{\rm cas}^{\rm H})^{\rm R}.
\end{equation}
This equivalence establishes that maximizing the determinant is directly linked to maximizing the trace, leading to a more tractable formulation of the original problem. 
Hence, the problem \textbf{(P1)} can be reformulated as 
\begin{subequations}
\begin{align}
& \max_{\stackrel{\boldsymbol{\Phi}_{\rm 1},\ldots,\boldsymbol{\Phi}_{\rm L}}{\mathbf{Q}}} ~ \mathrm{Tr}(\mathbf{H}_{\rm cas} \mathbf{Q}\mathbf{H}_{\rm cas}^{\rm H}),\label{trace_max}\\
& \text{ {s.t}} ~ \eqref{eq:power_constraint_cas} \text{~and~} \eqref{eq:Phi_constraint_cas}.
\end{align}\label{Trace_objective}
\end{subequations}
This transformation simplifies the problem from determinant to trace maximization, which is then solved through two subproblems: optimal precoding and RIS phase design.

\subsubsection{Optimal Precoding}\label{Sec: Precoding_cas}
In this subproblem, we optimize the precoding matrix $\mathbf{Q}$ for a given phase shift matrix $\mathbf{\Phi}$. The optimal $\mathbf{Q}$ is given by the eigenmode transmission \cite{tse2005fundamentals}. Taking SVD of $\mathbf{H}_{\rm cas} = \mathbf{U}_{\rm cas}\mathbf{S}_{\rm cas}\mathbf{V}_{\rm cas}^{\rm H}$, optimal $\mathbf{Q}$ can be selected as
\begin{equation}
    \mathbf{Q}^* = \mathbf{V}_{\rm cas}\left(p_1^*,\ldots,p_R^*\right)\mathbf{V}_{\rm cas}^{\rm H},\label{eq:svd-wf}
\end{equation}
where $R$ is the rank of $\mathbf{H}_{\rm cas}$, and $p_r^*$ denotes the optimal power allocated to the $r$-th data stream. The optimal power can be obtained using the water-filling strategy given in \cite{tse2005fundamentals} as
\begin{align}
    p_r^*=\left({1}/{\mu}-{\sigma_n^2}/{\lambda_r}\right)^+,\label{eq:WF}
\end{align}
where $\lambda_r$ is the $r$-th eigenvalue of $\mathbf{H}_{\rm cas}$ and $\mu$ is a Lagrangian variable that satisfies the power constraint i.e. $\sum_r p_r^* = P_t$.

\subsubsection{Optimal RIS phase shift}
Here, we aim to optimize $\mathbf{\Phi}$ given $\mathbf{Q}$. For easier exposition, we consider the case $L=2$ and later extend it to provide a generalized solution. Combining the BS--RIS$_1$ channel and the precoding, we can write $\Tilde{\mathbf{H}}_{0} = \mathbf{H}_{0} \mathbf{Q}^{\frac{1}{2}}$, and utilizing the cyclic property of trace, we can rewrite the trace expression in \eqref{Trace_objective} as
\begin{equation}
    \mathrm{Tr}\left(\mathbf{H}_{\rm cas}\mathbf{Q}\mathbf{H}_{\rm cas}^{\rm H}\right) = \mathrm{Tr}\left(\mathbf{H}_{2}^{\rm H}\mathbf{H}_{2}\mathbf{\Phi}_2\mathbf{H}_1\mathbf{\Phi}_1\Tilde{\mathbf{H}}_{0}\Tilde{\mathbf{H}}_{0}^{\rm H}\mathbf{\Phi}_1^{\rm H}\mathbf{H}_1^{\rm H}\mathbf{\Phi}_2^{\rm H}\right).\nonumber
\end{equation}
Taking SVD of $\Tilde{\mathbf{H}}_{0} = \Tilde{\mathbf{U}}_{0}\Tilde{\mathbf{S}}_{0}\Tilde{\mathbf{V}}_{0}^{\rm H}$, $\mathbf{H}_1 = \mathbf{U}_1\mathbf{S}_1\mathbf{V}_1^{\rm H}$, and $\mathbf{H}_{2} = \mathbf{U}_{2}\mathbf{S}_{2}\mathbf{V}_{2}^{\rm H}$, the above expression can be rewritten as 
{\small \begin{align}
    \mathrm{Tr}\left(\mathbf{U}_1^{\rm H}\mathbf{\Phi}_2^{\rm H}\mathbf{V}_2\mathbf{S}_2^2\mathbf{V}_2^{\rm H}\mathbf{\Phi}_2\mathbf{U}_1\mathbf{S}_1\mathbf{V}_1^{\rm H}\mathbf{\Phi}_1\Tilde{\mathbf{U}}_{0}\Tilde{\mathbf{S}}_{0}^2\Tilde{\mathbf{U}}_{0}^{\rm H}\mathbf{\Phi}_1^{\rm H}\mathbf{V}_1\mathbf{S}_1^{\rm H}\right), \label{eq: trace_ser_svd_exp}
\end{align}}
Using Von Neumann’s trace inequality \cite{VonNeumann_trace}, we can write
\begin{equation}
    \max\left[\mathrm{Tr}(\mathbf{H}_{\rm cas}\mathbf{Q}\mathbf{H}_{\rm cas}^{\rm H})\right] \leq \mathrm{Tr}\left(\Tilde{\mathbf{S}}_{0}^2\mathbf{S}_1^2\mathbf{S}_2^2\right). \label{eq:VN_bound}
\end{equation}
This is attributed to the fact that all the other matrices involved in \eqref{eq: trace_ser_svd_exp} are unitary. The maximum in \eqref{eq:VN_bound} is achieved upon setting the RIS phase shift matrices as $\boldsymbol{\Phi}_1 = \mathbf{V}_1\Tilde{\mathbf{U}}_{0}^{\rm H}$ and $ \boldsymbol{\Phi}_2^*  = \mathbf{V}_{2}\mathbf{U}_1^{\rm H}$. Notably, this configuration of the RIS phase shift matrices also satisfies the unitary matrix constraints $\boldsymbol{\Phi}_1\boldsymbol{\Phi}_1^{\rm H} = \mathbf{I}$ and $\boldsymbol{\Phi}_2\boldsymbol{\Phi}_2^{\rm H} = \mathbf{I}$. Thus, the optimal RIS phase shift matrices maximizing the trace, while satisfying the constraints, are 
\begin{subequations}
\begin{align}
    ~~& \boldsymbol{\Phi}_1^* = \mathbf{V}_1\Tilde{\mathbf{U}}_{0}^{\rm H},\label{eq:Phi1_cas_opt}\\
    \textit{\emph{and}} ~~& \boldsymbol{\Phi}_2^*  = \mathbf{V}_{2}\mathbf{U}_1^{\rm H}.\label{eq:Phi2_cas_opt}
\end{align}
\end{subequations}
For the above derived $\mathbf{Q}^*$ and $\boldsymbol{\Phi}_1^*$ and $\boldsymbol{\Phi}_2^*$, the determinant in the capacity equation \eqref{eq:MIMO_Capacity_cas} can be obtained as $|\mathbf{I}_{\rm K} + \mathbf{H}_{\rm cas}\mathbf{Q}^*\mathbf{H}_{\rm cas}^{\rm H}| = |\mathbf{I}_{\rm K} + \Tilde{\mathbf{S}}_{0}^2\mathbf{S}_1^2\mathbf{S}_{2}^2| = \prod_{k=1}^{\min(\Tilde{R}_0,R_1,R_2)} (1+\Tilde{\lambda}_{0k} \lambda_{1k} \lambda_{2k})$ where $\Tilde{R}_0$, and $R_l$ denote the ranks, and $\Tilde{\lambda}_{0k}$, and $\lambda_{lk}$'s, are the $k$-th eigenvalues of $\Tilde{\mathbf{H}}_{0}$, $\mathbf{H}_l$'s respectively. Using this, the optimal channel capacity can be written as
\begin{align}
    \mathrm{C}^* =\sum_{k=1}^{\min(R_0,R_1,R_2)}\log_2(1+\Tilde{\lambda}_{0k} \lambda_{1k} \lambda_{2k}),\label{eq:Cap_L=2}
\end{align}
The same approach generalizes to any $L$. From \eqref{eq:Phi1_cas_opt} and \eqref{eq:Phi2_cas_opt}, it is evident that each RIS phase matrix aligns the right singular vectors of the incoming channel with the left singular vectors of the outgoing channel. Accordingly, if $\mathbf{H}_{i-1}$ denotes the channel to the $i$-th RIS and $\mathbf{H}_{i}$ the channel from it, the optimal RIS configuration is obtained from their SVDs as
\begin{align}
    ~~&  \boldsymbol{\Phi}_i^* = \mathbf{V}_{i}\mathbf{U}_{i-1}^{\rm H}, \forall i = 2,\ldots, L\label{eq:Phii_cas_svd}\\
     \textit{\emph{and}} ~~& \boldsymbol{\Phi}_1^* = \mathbf{V}_{1}\Tilde{\mathbf{U}}_{0}^{\rm H}, \label{eq:Phi1_cas_svd}
\end{align}
where, $\mathbf{U}_{i-1}$ and $\mathbf{V}_{i}$ are the left and right eigenvectors of the links $\mathbf{H}_{i-1}$ and $\mathbf{H}_{i}$ respectively, and $\Tilde{\mathbf{U}_0}$ is the left singular eigenvector of  $\Tilde{\mathbf{H}}_{0}$. Thus extending \eqref{eq:Cap_L=2}, the optimal capacity for general $L$ can be expressed in closed form as 
{\small \begin{equation}
    \mathrm{C}^* = \sum_{k=1}^{R} \log_2\left(1+ \tilde{\lambda}_{0k}\prod_{l=1}^{L}\lambda_{lk}\right),\label{eq:Cap_svd_L}
\end{equation}}
where $R = \min(\Tilde{R}_0,R_1,\ldots,R_L)$, $\tilde{\lambda}_{0k}$ and $\lambda_{ik}$'s denote the $k$-th eigenvalue of the channels $\Tilde{\mathbf{H}}_0$ and $\mathbf{H}_i$ respectively.

It is noteworthy that the optimal RIS phase shifts in \eqref{eq:Phii_cas_svd}–\eqref{eq:Phi1_cas_svd} and the precoding matrix in \eqref{eq:svd-wf} are obtained in closed form. Although the optimal precoder and RIS matrices are interdependent, \eqref{eq:Phii_cas_svd} shows that $\boldsymbol{\Phi}_i$ for $i \geq 2$ depends only on $\mathbf{H}_i$, allowing a one-shot solution for $\{\boldsymbol{\Phi}_i\}_{i=2}^L$. Further the optimal solution for $\mathbf{\Phi}_1$ and $\mathbf{Q}$ can be obtained by solving \eqref{eq:Phi1_cas_svd} and \eqref{eq:svd-wf} iteratively. Algorithm \ref{Alg: cas+svd} outlines the step-by-step procedure to determine the optimal solution of \textbf{(P1)}.

\begin{algorithm}[!h]
\KwInput{$P_t$, $\epsilon$, $\mathbf{H}_l$   for $\forall l = 0,1,\ldots,L$} 
\KwOutput{$\mathbf{Q}^*$ and $\mathbf{\Phi}_l^* ~ \forall l=1,\ldots,L$}
Evaluate SVDs $\mathbf{H}_l = \mathbf{U}_l\mathbf{S}_l\mathbf{V}_l^{\rm H},~~\forall l$ and set $\boldsymbol{\Phi}_l^* = \mathbf{V}_{l}\mathbf{U}_{l-1}^{\rm H},~~\forall l = 2,\ldots,L.$\\
\KwInit{$\mathbf{Q}_0=\mathbf{I}_{M}$, $\boldsymbol{\Phi}_{1,0} = \mathbf{I}_{N_1}$ and $k=1$.}
\SetKwRepeat{Repeat}{Repeat}{Until:}
\Repeat{$|\mathrm{C}(\mathbf{\Phi}_{1,k},\mathbf{Q}_k,\mathbf{\Phi}_{l}^*)-\mathrm{C}(\mathbf{\Phi}_{1,k-1},\mathbf{Q}_{k-1},\mathbf{\Phi}_{i}^*)| \leq \epsilon$}
{
Calculate $\Tilde{\mathbf{H}}_{0} = \mathbf{H}_0\mathbf{Q}_{k-1}^{\frac{1}{2}}$\\
Evaluate the SVD $\Tilde{\mathbf{H}}_{0} = \Tilde{\mathbf{U}}_{0}\Tilde{\mathbf{S}}_{0}\Tilde{\mathbf{V}}_{0}^{\rm H}$ and set $\boldsymbol{\Phi}_{1,k} = \mathbf{V}_{1}\Tilde{\mathbf{U}}_{0}^{\rm H}$\\
Compute $\mathbf{H}_{{\rm cas},k} = \mathbf{H}_{\rm L}\mathbf{\Phi}_{\rm L}^*\mathbf{H}_{\rm L-1} \mathbf{\Phi}_{\rm L-1}^* \dots \mathbf{H}_{2} \mathbf{\Phi}_2^* \mathbf{H}_1 \mathbf{\Phi}_{1,k} \mathbf{H}_{0}$\\
Evaluate the SVD $\mathbf{H}_{{\rm cas},k} = \mathbf{U}_{\rm cas}\mathbf{S}_{\rm cas}\mathbf{V}_{\rm cas}^{\rm H}$.\\
Update $\mathbf{Q}_{k}$ according to \eqref{eq:svd-wf}.\\
$k \longleftarrow k+1$
}
Set $\mathbf{\Phi}_1^*=\mathbf{\Phi}_1$, $\mathbf{\Phi}_l^*=\mathbf{\Phi}_l ~\forall l=2,\ldots,L$ and $\mathbf{Q}^*=\mathbf{Q}_k$
\caption{Cascaded RIS: Capacity maximization with SVD-WF based precoding}
\label{Alg: cas+svd}
\end{algorithm}
For SVD–WF precoding, the optimal solution requires iterative updates (Algorithm \ref{Alg: cas+svd}). Replacing the precoder with UPA reduces the computational complexity by allocating power uniformly across all $M$ antennas irrespective of channel states. This decouples $\mathbf{Q}$ and $\mathbf{\Phi}_0$, enabling a one-shot closed-form solution for both, as presented in the following theorem.

\begin{theorem}\label{Theo: UPA_cas}
    For UPA precoding, i.e. $\mathbf{Q} = \frac{P_t}{M}\mathbf{I}_M$, the optimal phase shift matrices for cascaded RISs are given by
    \begin{align}
        \boldsymbol{\Phi}_i^*&= \mathbf{V}_{i}\mathbf{U}_{i-1}^{\rm H}, \forall i = 1,\ldots, L\label{eq:Phi_upa_cas}
    \end{align}
    and the achievable capacity is given by
    {\small \begin{align}
        \mathrm{C}^*& = \sum_{k=1}^{R} \log_2\left(1+ \tfrac{P_t}{M}\prod_{l=0}^L\lambda_{lk}\right),\label{eq:Capacity_upa_cas}
    \end{align}}
    where $R=\min(R_0,R_1,\ldots,R_L)$, $R_l$ and $\lambda_{lk}$ are the rank and the $k$-th  eigenvalue of $\mathbf{H}_l$.
\end{theorem}

\begin{proof}
For UPA precoding, where $\mathbf{Q} = \tfrac{P_t}{M}\mathbf{I}_M$, we have $\Tilde{\mathbf{H}}_{0} = \mathbf{H}_0 \mathbf{Q}^{1/2} = \sqrt{\tfrac{P_t}{M}}\mathbf{H}_0$. Hence, $\Tilde{\mathbf{H}}_{0}$ differs from $\mathbf{H}_0$ only by a scalar 
factor which leaves the singular vectors of 
$\mathbf{H}_0$ unchanged, while only scaling its singular values. Since the optimal RIS phase design depends only on the singular vectors, the optimal $\mathbf{\Phi}_1$ given in \eqref{eq:Phi1_cas_svd} becomes 
\begin{align}
    \mathbf{\Phi}_1^* = \mathbf{V}_1 \mathbf{U}_0^{\rm H},
\end{align}
The remaining RIS phase matrices $\mathbf{\Phi}_i, \, \forall i = 2,\ldots,L$, remain independent of the precoding matrix and thus follow directly from \eqref{eq:Phii_cas_svd}. This establishes the general optimal RIS configuration given in \eqref{eq:Phi_upa_cas} for UPA precoding. 
Furthermore, for UPA precoding, the eigenvalues of $\Tilde{\mathbf{H}}_0$ are scaled versions of 
those of $\mathbf{H}_0$, i.e., $ \Tilde{\lambda}_{0k} = \tfrac{P_t}{M} \lambda_{0k}$, where $\lambda_{0k}$ denotes the $k$-th eigenvalue of $\mathbf{H}_0$. Substituting this into  \eqref{eq:Cap_svd_L},  we obtain the capacity for UPA precoding as given in \eqref{eq:Capacity_upa_cas}. 
\end{proof}

\begin{remark}\label{rem:UPAvsSVD}
With optimally configured cascaded RISs, we numerically observe that UPA achieves performance comparable to SVD-WF based precoding: consider the effective channel matrix from the BS to users via RIS as $\mathbf{H}_{\mathrm{eff}}$ whose $(m,k)$-th element can be written as  
\begin{equation}  
    [\mathbf{H}_{\mathrm{eff}}]_{m,k} = \sum_{n=1}^{N} (\mathbf{H}_2)_{m,n} e^{j\phi_n} (\mathbf{H}_1)_{n,k},  
\end{equation}  
where $\phi_n$ denotes the phase shift provided by the $n$-th RIS element. As $N \to \infty$, the channel hardening effect occurs which ensure that the random fluctuations get averaged out, yielding  
\begin{equation}  
    \lim_{N\to\infty}\frac{1}{N}\mathbf{H}_{\mathrm{cas}}\mathbf{H}_{\mathrm{cas}}^{H} \xrightarrow{} \mathbb{E}\!\left[\frac{1}{N}\mathbf{H}_{\mathrm{cas}}\mathbf{H}_{\mathrm{cas}}^{H}\right]. 
\end{equation}  
Consequently, the variance of eigenvalues of $\mathbf{H}_{\mathrm{eff}}$ diminishes with $N$, implying that all non-zero eigenvalues become nearly equal, i.e. $\lambda_1 \approx \lambda_2 \approx \cdots \approx \lambda_R$. 
This equalization implies that the effective channel behaves like a scaled identity matrix. Thus, all eigenmodes are allocated equal power as $ p_1 = p_2 = \cdots = p_R = \frac{P_t}{M}$. Therefore, the waterfilling solution in SVD-WF based algorithm reduces to UPA with covariance matrix
\begin{equation}
    \mathbf{Q} = \mathbf{V} \operatorname{diag}(p_1^*,\ldots,p_R^*) \mathbf{V}^{\rm H}= \tfrac{P_t}{M}\mathbf{I}_{\rm M},
\end{equation}
since $\mathbf{V}\mathbf{V}^{\rm H} = \mathbf{I}$.  
Hence, with a sufficiently large and optimized RIS, UPA becomes asymptotically equivalent to SVD-based precoding—a practically relevant observation since RISs often employ large arrays. This is analogous to large-scale MIMO systems, where channel hardening and eigenvalue equalization render UPA nearly optimal as array size increases.
\end{remark}  

Note that regardless of the precoder $\mathbf{Q}$, the RIS functions as a phase-matching layer that aligns dominant channel eigenmodes, maximizing the SNR gain in the cascaded channel.

\subsection{Special Case 1: Single RIS System}
The results derived for the cascaded RIS arrangement can be directly reduced to the single-RIS case i.e. $L=1$. For this case, we present solutions for both type of precoding i.e. SVD-WF and UPA in the following table.
{\small \begin{table}[!h]
\centering
\caption{Optimal RIS Phase Shift and Capacity for $L=1$.}
\label{Table:SVD-UPA_L=1}
\resizebox{\columnwidth}{!}{
\begin{tabular}{|c|c|c|}
\hline
\textbf{Precoding} & \textbf{Optimal RIS Phase Shift} & \textbf{ Optimal Capacity} \\
\hline
SVD-WF & $ \boldsymbol{\Phi}_1^*\! =\! \mathbf{V}_1\mathbf{\Tilde{U}}_0^{\rm H}$ 
& $\sum_{k=1}^{\min(\Tilde{R}_0,R_1)}\log_2\left(1+\Tilde{\lambda}_k\sigma_k\right)$ \\
\hline
UPA & $\boldsymbol{\Phi}_1^* = \mathbf{V}_1\mathbf{U}_0^{\rm H}$ 
& $\sum_{k=1}^{\min(R_0,R_1)}\log_2\left(1+ \frac{P_t}{M}\lambda_k\sigma_k\right)$ \\
\hline
\end{tabular}
}
\end{table}}

In SVD-WF precoding, $\boldsymbol{\Phi}_1$ and $\mathbf{Q}$ needs to be iteratively obtained similar to Algorithm \ref{Alg: cas+svd} for their jointly optimal values. 

\subsection{Special Case 2: Diagonal Passive RIS}
For passive D-RIS where its elements are not connected, the phase shift matrix has unit modulus constraint. Here, directly obtaining the optimal phase shift matrix is computationally complex, e.g., refer to the algorithm presented in \cite{zhang2021capacity}. Interestingly, our proposed solution in Algorithm \ref{Alg: cas+svd} can yield a near-optimal phase shift matrix $\mathbf{\Tilde{\Phi}}$ for D-RIS. Specifically, we obtain $\Tilde{\mathbf{\Phi}}$ by projecting the optimal unitary solution $\mathbf{\Phi}^*$ onto the feasible set of diagonal unit-modulus matrices. This is achieved by minimizing the Frobenius distance
\begin{equation}
\min_{\mathbf{\tilde{\Phi}} = \mathrm{diag}(e^{i\theta_j})} \| \mathbf{\Phi^*} - \mathbf{\tilde{\Phi}} \|_F^2.\label{eq:Frob_norm}
\end{equation}

Since $\mathbf{\tilde{\Phi}}$ is diagonal, we can write the argument in \eqref{eq:Frob_norm} as 
\begin{equation}
\| \mathbf{\Phi^*} - \mathbf{\tilde{\Phi}} \|_F^2 = \sum_{j=1}^n |\mathbf{{\Phi}}^*_{jj} - e^{i\theta_j}|^2 + \sum_{j \neq k} |\mathbf{{\Phi}}^*_{jk}|^2,    
\end{equation}
which is minimized when $e^{i\theta_j}$ (a complex number on the unit circle) is closest to  $\mathbf{{\Phi}}^*_{jj}$. Thus, the near-optimal unit modulus diagonal phase shift matrix can be determined as 
\begin{equation}
    \mathbf{\tilde{\Phi}^*} = \mathrm{diag}\left( \frac{\mathbf{{\Phi}}^*_{11}}{|\mathbf{{\Phi}}^*_{11}|}, \frac{\mathbf{{\Phi}}^*_{22}}{|\mathbf{{\Phi}}^*_{22}|}, \dots, \frac{\mathbf{{\Phi}}^*_{NN}}{|\mathbf{{\Phi}}^*_{NN}|} \right).\label{eq:diagRIS}
\end{equation}

\section{Ergodic Capacity for cascaded RIS}
Here, we analyze the EC of $L$ cascaded RISs and derive closed-form approximations for (i) high-SNR regime and, (ii) the case of moderately large $N$ compared to $M$ and $K$.

\begin{theorem}\label{theo: EC_cas}
For optimally configured cascaded RIS-based MIMO system with UPA precoding under Rayleigh channel, the ergodic capacity $\mathbb{E}[C_{\rm cas}^*]$ can be approximated as 
\begin{equation}
\mathrm{EC}_{\rm cas} = \!\sum_{k=1}^{R}\!\log_2\!\!\left(1+\tfrac{P_t}{M}\mathcal{M}_1\!\right)
-\!\frac{(\frac{P_t}{M})^2\!\left(\!\mathcal{M}_2-\!\mathcal{M}_1^2\!\right)}{2\log_2\!\left(1+\!\mathcal{M}_1\right)^{\!2}},
\label{eq:EC_cas_theo2}
\end{equation}
where $\mathcal{M}_n=\prod_{l=0}^{L}\!\mathbb{E}[\lambda_{l_k}^n]$, $R= \min(R_0,R_1,\ldots,R_L)$ such that $R_l$ is the rank of $\mathbf{H}_l$, and $\lambda_{lk}$ is the $k$-th eigenvalue of $\mathbf{H}_l$. The $n$-th moment of the eigenvalue of a Wishart matrix is obtained using \cite{nguyen2023depth} as
    \begin{align}
        \mathbb{E}[\lambda^n] = \tfrac{\Gamma(n + \alpha + 1)}{r \Gamma(\alpha + 1)} \sum_{k=0}^{r-1} \tfrac{1}{k!} \tfrac{d^k}{dh^k} \left[ \tfrac{{}_2F_1(a,b;c;z)}{(1 - h)^{-n}(1 + h)^{1 + \alpha + n}} \right],\label{eq: Moment}\nonumber
    \end{align}
such that $\alpha$ is the difference between the matrix dimensions and \( {}_2F_1(a,b;c;z) \) is the Gaussian hypergeometric function with parameters 
$$      
a = \tfrac{1 + \alpha + 2}{2},~ 
b = 1 + \tfrac{\alpha + n}{2}, ~
c = 1 + \alpha,~\text{and} ~
z = \tfrac{4h}{(1 + h)^2}. 
$$
\end{theorem}

\begin{proof}
EC for optimally configured cascaded RIS can be obtained using \eqref{eq:Capacity_upa_cas} as   
\begin{equation}
     \mathbb{E}[C_{\rm cas}^*] = \sum_{k=1}^{R} \mathbb{E}\left[\log_2\left(1+ \tfrac{P_t}{M}\prod_{l=0}^L\lambda_{lk}\right)\right].
\end{equation}
Using the second-order Taylor series approximation of $\log_2(1 + x) $, we approximate the above expectation as $\mathrm{EC}_{\rm cas} =$
{\small \begin{align*}
     \sum_{k=1}^{R} \log_2\left(1 + \mathbb{E}\left[\tfrac{P_t}{M}\prod_{l=0}^L\lambda_{lk}\right]\right) 
 - \frac{0.5\mathbb{V}\text{ar}(\tfrac{P_t}{M}\prod\limits_{l=0}^L\lambda_{lk})}{ \log 2 \left(1 + \mathbb{E}\left[\tfrac{P_t}{M}\prod\limits_{l=0}^L\lambda_{lk}\right]\right)^2}.
\end{align*}}
Since $\mathbf{H}_l$'s are independent, the corresponding $\lambda_{l_k}$'s are also independent and thus $\mathrm{EC}_{\rm cas}$ can be rewritten as \eqref{eq:EC_cas_theo2}.
\end{proof}

Although the expression in Theorem \ref{theo: EC_cas} is in closed form, its evaluation is computationally demanding due to higher-order moments and hypergeometric terms, and it does not offer insights into how capacity scales with $N$, $M$, and $K$. To address this, the following theorem presents a tractable high-SNR approximation.

\begin{theorem}\label{theo: EC_cas_highSNR}
At high SNR, EC with optimally configured cascaded RISs and UPA precoding is obtained as $\mathrm{EC}_{\rm HS} =$
\begin{equation}
 \!\min(K,M)\log_2\!\!\left(\tfrac{P_t}{M}\right)
+\!\sum_{l=0}^{L}\!R_l + \log_2|\mathbf{\Sigma}_l|
+\!\tfrac{1}{\ln 2}\!\sum_{i=1}^{R_l}\!\psi\!\!\left(\frac{\nu_l\!-\!i\!+\!1}{2}\right)
\label{eq:Capacity_highSNR}
\end{equation}
where $\mathbf{\psi(.)}$ is a digamma function, and $\mathbf{H}_l\mathbf{H}_l^{\rm H} \sim \mathcal{W}\left(\nu_l,\mathbf{\Sigma}_l\right)$.
\label{Theo: HIgh SNR}
\end{theorem}
\begin{proof}
    At high SNR, capacity in \eqref{eq:Capacity_upa_cas} can be approximated as
    {\small \begin{align}
        \log_2 \left|\mathbf{I}_K + \frac{P_t}{M}\prod_{l=0}^L\mathbf{S}_l^2\right| & \stackrel{(a)}{=}  \log_2\left|\frac{P_t}{M}\prod_{l=0}^L\mathbf{S}_l^2\right|_+,\nonumber\\
         &= \log_2\left|\frac{P_t}{M}\right|_+\left|\mathbf{S}_0^2\right|_+\left|\mathbf{S}_1^2\right|_+\ldots\left|\mathbf{S}_L^2\right|_+,\nonumber \\
         &= \min(K,M)\log_2\left(\frac{P_t}{M}\right) + \sum_{l=0}^{L} \log_2 \left| \mathbf{S}_l^2\right|_+.\nonumber
    \end{align}}
    In Step (a), the pseudo-determinant (denoted as $|\cdot|_+$) is used to handle rank-deficient cases where $\det(\mathbf{I}+\mathbf{X}) \approx \det(\mathbf{X})$ fails, as $\det(\mathbf{X})$ becomes zero for non–full-rank $\mathbf{X}$.
    Note that, $\det(\mathbf{I}+\mathbf{X})$ remains strictly positive since the addition of the identity shifts the zero eigenvalues of $\mathbf{X}$ to one. By considering pseudo-determinant, defined as the product of only the nonzero eigenvalues, we avoid this degeneracy and obtain a finite and meaningful quantity that closely reflects the behavior of $\det(\mathbf{I}+\mathbf{X})$. EC can now be obtained as
    \begin{align}
     \mathrm{EC}_{\rm HS} = \min(K,M)\!\log_2\!\left(\frac{P_t}{M}\right)\! +\! \sum_{l=0}^{L} \mathbb{E}\left[\log_2 \left| \mathbf{S}_l^2\right|_+\right].\label{eq: EC_high_SNR}
    \end{align}
    Using the known expectation of log-determinant of a Wishart matrix \cite{nguyen2023depth}, EC is obtained as in \eqref{eq:Capacity_highSNR}. 
\end{proof}

\begin{corollary}\label{cor: EC_cas_highSNR}
At high SNR and independent fading, with a moderately large number of RIS elements $N$ compared to $M$ and $K$,  EC expression given in \eqref{eq:Capacity_highSNR} can be simplified as 
\begin{align}
      \tilde{\rm EC}_{\rm HS}\! = & \min(K,M)\log_2\!\left(\tfrac{P_t}{M}\right) + 2\min(K,M)\log_2(N) \nonumber\\
     & + M(L-1)\log_2(N). \label{eq: EC_cf_highSNR}
\end{align}
\end{corollary}

\begin{proof}
Since  $\mathbf{H}_l\mathbf{H}_l^{\rm H} \sim \mathcal{W}\left(\nu_l,\mathbf{I}\right)$ for independent fading, the term $\log_2|\mathbf{\Sigma}_i|$ reduces to zero. The parameter $\nu_l$ represents DoF of Wishart matrix which here is $N$. 
Further, the digamma function behaves logarithmically  for a large argument \cite{Digamma_Log}. 
Thus, for $N \gg M,K$, we have
\begin{align*}
R_l+\sum_{i=1}^{R_l}\frac{\psi(\frac{\nu_i-i+1}{2})}{\ln 2} \approx R_l+R_l\log_2(\tfrac{N}{2})=R_l\log_2(N).    
\end{align*}
Combining these arguments, we can simplify \eqref{eq:Capacity_highSNR} to \eqref{eq: EC_cf_highSNR}.
\end{proof}

All results in Theorems \ref{theo: EC_cas}, \ref{theo: EC_cas_highSNR}, and Corollary \ref{cor: EC_cas_highSNR} directly extend to the single RIS case by setting $L=1$.

\begin{remark}\label{rem:N_opt}
The number of elements $N_{\rm req}$ required per RIS for achieving a target capacity $C_{\rm tar}$ at high SNR can be easily obtained from Corollary \ref{cor: EC_cas_highSNR} as
\begin{equation}
    \log_2(N_{\rm req}) = {\frac{C_{\rm tar} - \min(K,M)\log_2\tfrac{P_t}{M}}{2\min(K,M)+M(L-1)}}.\label{eq:N_opt}
\end{equation}
This relation offers key insights for designing cascaded RIS-aided MIMO systems. The term $\min(\!K,M\!)\log_2\tfrac{P_t}{M}$ represents the achievable capacity of a $K \times M$ MIMO system with UPA precoding  solely through the transmit power $P_t$ at high SNR . Thus, the numerator $C_{\rm tar} \!-\! \min(\!K,M\!)\! \log_2\tfrac{P_t}{M}$ captures the additional capacity that is required beyond this $K\!\times \!M$ MIMO system. 
As the target capacity increases, the required number of RIS elements grows exponentially, reflecting the fact that  more passive beamforming for larger array gain is required to bridge the gap between $C_{\rm tar} $ and $\min(K,M)\log_2\tfrac{P_t}{M}$. This also highlights the power--RIS tradeoff: increasing $P_t$ reduces $N_{\rm req}$ required for the same $C_{\rm tar}$, emphasizing the complementary roles of active and passive resources. The denominator $2\min(K,M) \!+\! M(L\!-\!1)$ captures the effective spatial DoF of the system. Increasing $M$ or $K$ improves the multiplexing capability, reducing the RIS dimension for a given $C_{\rm tar}$. Similarly, having multiple RISs ($L\!>\!1$) in cascade increases the denominator, so fewer elements per RIS are needed to achieve the same rate. 
For instance, in a single RIS-aided MIMO system with spatial DoF $\min(K,M)$, the received SNR with RIS array gain is $N_{\rm req}^2 \tfrac{P_t}{M}$. Accordingly, the achievable capacity with $N_{\rm req}$ elements is
\begin{align*}
     C_{\rm tar}&=\min(K,M)\log_2\left(N_{\rm req}^2\tfrac{P_t}{M}\right),\\&=\min(K,M)\log_2\left(\tfrac{P_t}{M}\right)+2\min(K,M)\log_2(N_{\rm req}),
\end{align*}
which is consistent with \eqref{eq:N_opt}. 
Overall, this expression clearly characterizes the tradeoff between transmit power, RIS size, and achievable capacity, providing practical guidance for RIS deployment in high-SNR MIMO systems.
\end{remark}

\section{Numerical Results}
\begin{figure*}[!ht]
     \centering
         \centering
         \includegraphics[width=0.325\textwidth]{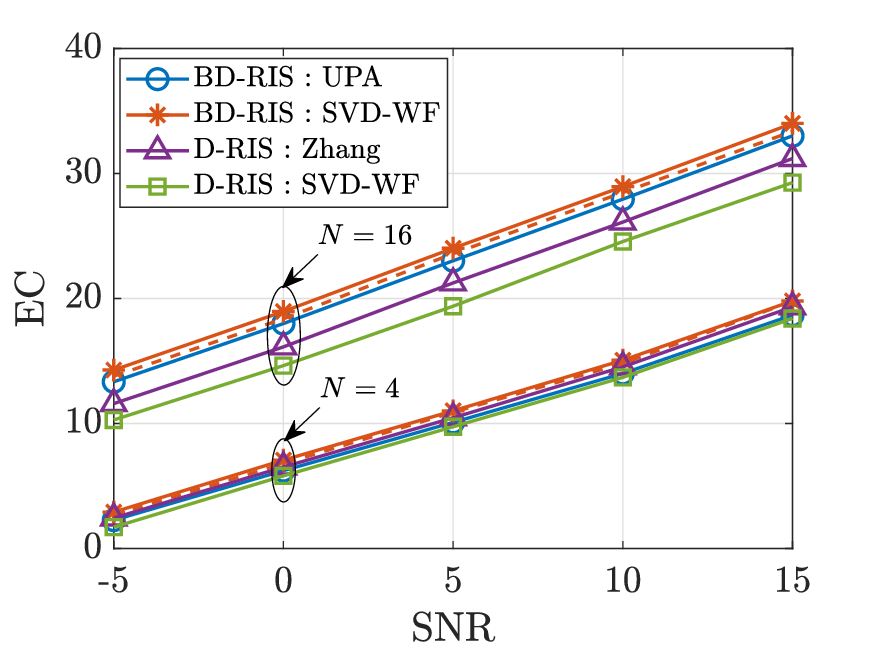}
         \includegraphics[width=0.325\textwidth]{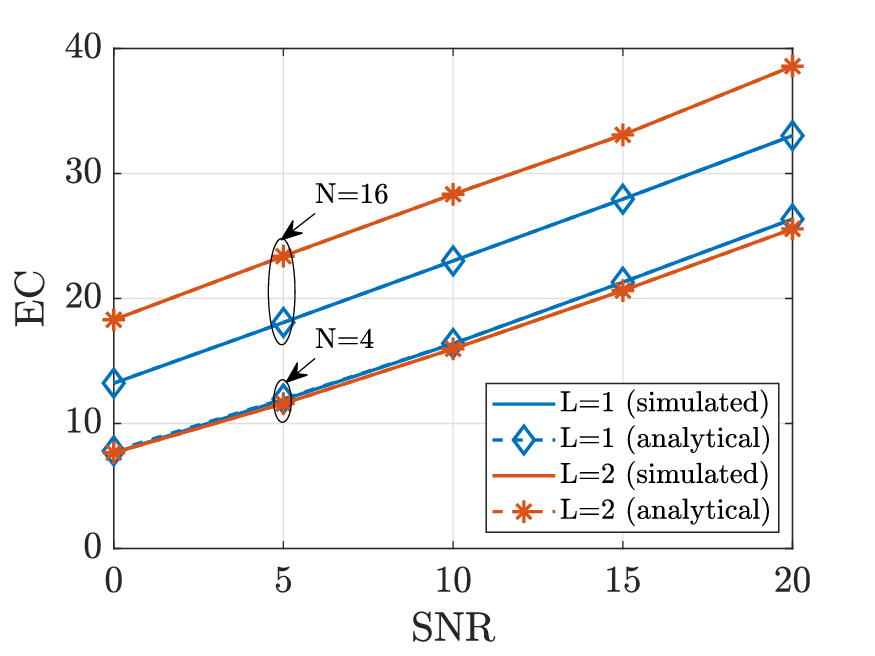}
         \includegraphics[width=0.325\textwidth]{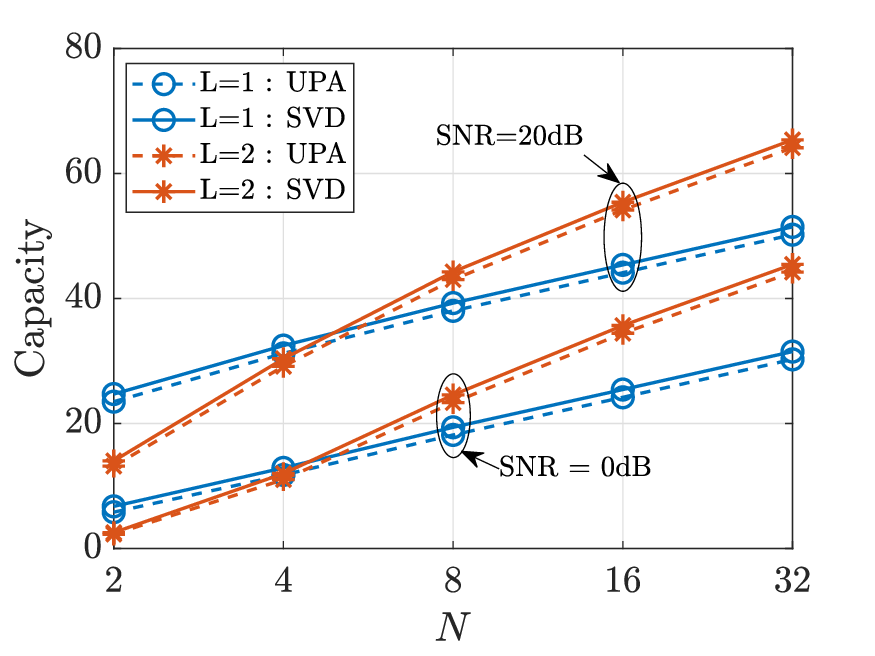}\\\vspace{-2mm}
         \text{(a)\hspace{5.6cm}(b)\hspace{5.6cm}(c)}\vspace{-1.5mm}
        \caption{\small (a) EC vs. SNR for single D-RIS and BD-RIS, dashed lines represent Monte Carlo result for BD-RIS with SVD-WF (b) EC vs. SNR for cascaded BD-RISs, (c) EC vs. $N$ for cascaded BD-RISs.}\vspace{-2mm}
        \label{fig:EC_results}
\end{figure*}

This section presents numerical results validating EC of the proposed cascaded BD-RIS–aided MIMO system and comparing it with equivalent single BD/D-RIS setups.
Unless stated otherwise, we consider $L=2$, $M=K=4$, $P_t=10$, $\sigma_n^2=1$, and independently fading channels with $\mathbf{H}_l \sim \mathcal{CN}(\mathbf{0}, \mathbf{I})$.

Fig. \ref{fig:EC_results}(a) provides numerical validation of the derived EC for a single BD and D-RIS aided MIMO system achievable via proposed solutions. 
The capacities given in Table \ref{Table:SVD-UPA_L=1} for SVD-WF and UPA precoding match closely with the simulated capacities. This verifies the accuracy of proposed solution for configuring RIS and the fact that optimally configured RIS with UPA and SVD-WF yield similar performance as highlighted in Remark \ref{rem:UPAvsSVD}. Besides, the capacity achievable through the near-optimal solution presented in \eqref{eq:diagRIS} is close to Zhang’s optimal algorithm \cite{zhang2021capacity} for D-RIS case. Interestingly, BD-RIS with UPA outperforms D-RIS as $N$ increases, demonstrating its efficiency and lower complexity compared to the computationally intensive D-RIS implementation \cite{zhang2021capacity}. 

Fig.~\ref{fig:EC_results}(b) shows that EC achievable via Algorithm \ref{Alg: cas+svd} matches exactly with the simulated results for both $L=1$ and $L=2$.
For fair comparison, we consider a single RIS with $N$ elements, while each of the $L$ cascaded RISs has $\tfrac{N}{L}$ elements. 
The figure highlights an important fact that EC of single and cascaded RIS systems are almost equal at $N=4$, whereas the cascaded arrangement provides superior  EC compared to the single-RIS case as $N$ increases to $16$. This is attributed to the fact that SNR gain improves quadratically with the number of optimally configured RIS elements $N$ \cite{basar2019wireless}, giving the overall gain of $\left(\tfrac{N}{L}\right)^{2L}$ in $L$ cascaded RISs. Thus, for  relatively small $N$, cascaded deployments
provide gains lesser than $N^2$ and is expected to perform poorly compared to single RIS (as can be seen in Fig. \ref{fig:EC_results}(c)). For parameters considered in Fig. \ref{fig:EC_results}(b), both  arrangements exhibits similar performance. 
However, the SNR gain with cascaded RIS will be significantly larger compared to single RIS when $N \gg L$, verifying the superior performance of the cascaded arrangement. 
Nevertheless, the crossover point between performances of single and cascaded RIS arrangements remain invariant to SNR as the relative SNR gains is determined solely by the inherent power-scaling laws rather than $P_t$. Specifically, the crossover occurs at $N^{2} = \left(\tfrac{N}{L}\right)^{2L}$ or equivalently at $N = L^{\frac{L}{L-1}}$. This can be verified from Fig. \ref{fig:EC_results}(c) wherein the crossover point  occurs approximately at $N=4$ for $L=2$ at $0$ and $20$ dB SNRs.

\begin{figure}[!ht]
         \centering\includegraphics[width=0.35\textwidth]{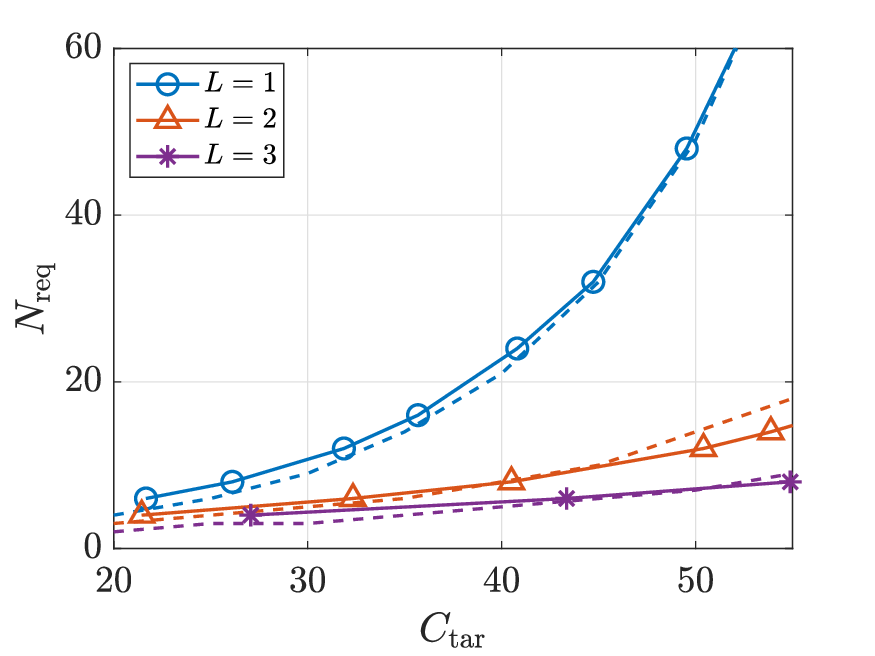}\vspace{-2.5mm}
        \caption{\small Number of RIS elements reequired to achieve a target capacity. The solid and dashed curves represent the analytical and simulation results, respectively. }
        \label{fig:N_opt}
\end{figure}
Fig. \ref{fig:N_opt} verifies that the number of RIS elements $N_{\rm req}$ required to achieve target capacity $C_{\rm tar}$ given in Remark \ref{rem:N_opt} closely match simulations.
It can be observed that the single RIS system requires a significantly large $N_{\rm req}$ compared to cascaded RISs as $C_{\rm tar}$ increases. Thus, it is evident that cascaded RIS deployments provide a clear advantage over the single RIS case.  
Furthermore, elements per RIS, $N_{\rm req}$, is much lesser for $L=3$ than $L=2$, leading to fewer total elements overall as more RISs are cascaded.


\section{Conclusion}
This paper presents jointly optimal closed-form solutions for cascaded BD-RISs with SVD–WF and UPA precoding schemes, demonstrating that the low-complexity UPA precoder achieves performance comparable to SVD–WF. The solution extends naturally to D-RISs with near-optimal, low-complexity performance. An expression for approximate ergodic capacity under UPA is also derived and simplified for high-SNR and moderately large $N$. The results highlight the superior capacity gains of cascaded-RISs over single RIS, providing practical design insights on the number of elements required per RIS for the cascaded setting. This work can be further extended to evaluate the performance of BD-RIS in alternative deployments, such as parallel arrangement.

\bibliographystyle{IEEEtran}

\end{document}